\documentclass[review]{elsarticle}
\usepackage{hyperref}
\usepackage{amsmath,amssymb}
\usepackage{mathtools}
\usepackage{mathtools}
\usepackage{algorithm}
\usepackage{algorithmicx}
\usepackage[noend]{algpseudocode}
\usepackage{algpascal}
\usepackage{color}
\usepackage{graphicx}
\usepackage{enumerate}
\usepackage{colortbl}
\usepackage{vmargin}
\usepackage{pdfpages}

\journal{}

\usepackage{tikz}

\newcommand{\sm}{\setminus}

\newtheorem{theorem}  {Theorem}          
\newtheorem{lemma}    {Lemma}             
\newtheorem{corollary}  {Corollary} 
\newtheorem{remark}   {Remark}          
\newtheorem{obs}      {Observation}

\usepackage{amsmath}
\usepackage{amssymb}
\usepackage{color}
\definecolor{red}{RGB}{255,0,0}
\definecolor{blue}{RGB}{0,0,255}
\definecolor{green}{RGB}{0,255,0}

\newenvironment{proof}{\noindent{\bf Proof:}}
{\begin{flushright}$\square$\end{flushright}}

\newcommand {\abs}[1]  {\left\vert#1\right\vert}
\newcommand {\set}[1]  {\left\{#1\right\}}
\newcommand {\vect}[1] {{\mathbf{#1}}}
\newcommand {\defined} {\stackrel{def} {=}}

\newcommand {\nph}     {\textsc{NP}\textrm{-hard}}
\newcommand {\fpt}     {\sf FPT}
\newcommand {\xp}      {\sf XP}
\newcommand {\ETH}     {\sf ETH}

\newcommand {\bigoh}   {{\mathcal O}}

\newcommand {\runningtitle}[1] {\vspace{0.5ex}\noindent{\textbf{\boldmath #1:}}}

\DeclareMathOperator{\cs}{cs}

\newcommand{\BB}{\mathcal B}

\newcommand{\VV}{\mathcal V}
\newcommand{\CC}{\mathcal C}
\newcommand{\II}{\mathcal I}
\newcommand{\LL}{\mathcal L}

\newcommand{\maxcut} {\textsc{MaxCut}}
\newcommand{\cutsize}[1] {\cs(#1)}
\newcommand{\cw} {{\sf cw}}
\newcommand{\bw} {{\sf bw}}
\newcommand{\BW} {{\sf BW}}

\begin{document}
\begin{frontmatter}

\title{On the Maximum Cardinality Cut Problem in Proper Interval Graphs and Related Graph Classes}

\author[bu]{Arman Boyac{\i}\fnref{tubitak}}
\author[bu]{Tinaz Ekim\fnref{tubitak}}
\author[telhai,hit]{Mordechai Shalom \corref{corr}}

\address[bu]{Department of Industrial Engineering, Bo\u{g}azi\c{c}i University, Istanbul, Turkey}
\address[telhai]{TelHai College, Upper Galilee, 12210, Israel}
\address[hit]{Holon Institute of Technology, Israel}

\cortext[corr]{Corresponding author}

\begin{abstract}
Although it has been claimed in two different papers that the maximum cardinality cut problem is polynomial-time solvable for proper interval graphs, both of them turned out to be erroneous. 
In this paper, we give $\fpt$ algorithms for the maximum cardinality cut problem in classes of graphs containing proper interval graphs and mixed unit interval graphs when parameterized by some new parameters that we introduce. 
These new parameters are related to a generalization of the so-called bubble representations of proper interval graphs and mixed unit interval graphs and to clique-width decompositions. 
\end{abstract}


\end{frontmatter}


\section{Introduction}\label{sec:intro}
A \emph{cut} of a graph $G=(V(G),E(G))$ is a partition of $V(G)$ into two subsets $S, \bar{S}$ where $\bar S=V(G)\setminus S$. 
The \emph{cut-set} of $(S,\bar{S})$ is the set of edges of $G$ having exactly one endpoint in $S$. 
\emph{The maximum cardinality cut problem ($\maxcut$)} is to find a cut with a maximum size cut-set, of a given graph.

$\maxcut$ remains $\nph$ when restricted to the following graph classes: chordal graphs, undirected path graphs, split graphs, tripartite graphs, co-bipartite graphs \cite{Bodlaender00onthe}, unit disk graphs \cite{DK2007} and total graphs \cite{Guruswami1999217}. 
On the positive side, it was shown that $\maxcut$ can be solved in polynomial-time in planar graphs \cite{hadlock1975finding}, in line graphs \cite{Guruswami1999217}, in graphs with bounded clique-width \cite{fomin2010algorithmic}, and the class of graphs factorable to bounded treewidth graphs \cite{Bodlaender00onthe}.
None of these results applies to proper interval graphs.

Polynomial-time algorithms for some subclasses of proper interval graphs (also known as indifference graphs) are proposed in \cite{Bodlaender2004} and in \cite{BES15-MaxCut-JOCOSpecialIssue}, for split indifference graphs and co-bipartite chain graphs, respectively. 
A polynomial-time algorithm for proper interval graphs is proposed in \cite{bodlaenderKN99MaxCutUnitInterval}. 
However, as pointed out in \cite{Bodlaender2004} this algorithm contains a flaw and may return sub-optimal solutions. 
A polynomial-time algorithm for proper interval graphs was also proposed by the authors of this work \cite{BES16-MaxCutProperInterval}. 
However, this algorithm too, is flawed as explained in detail in \cite{Kratochvil-UBubbleModel-2020}.
Consequently, the question of whether $\maxcut$ can be solved in polynomial time for proper interval graphs is open. 

In the same work (\cite{Kratochvil-UBubbleModel-2020}) the well known two dimensional bubble model for proper interval graph is extended to mixed unit interval graphs and a sub-exponential exact algorithm is given for this family of graphs.

As for the parameterized complexity of the problem in general graphs, $\maxcut$ when parameterized by the clique-width of the input graph is not in $\fpt$ unless the Exponential Time Hypothesis ($\ETH$) collapses \cite{fomin2010algorithmic}.

\runningtitle{Our Contribution}
In this work, we consider the parameterized complexity of the $\maxcut$ in depending on graph parameters related to clique-width decomposition.
We present $\fpt$ algorithms for the problem when parameterized with these parameters.
Our results imply $\fpt$ algorithms for mixed unit interval graphs.

In Section \ref{sec:BubblePartition} we introduce the notion of bubble partitions and the parameters independence number $\alpha$ and width of a bubble partition. 
This notion generalizes the two dimensional bubble model of proper interval and mixed unit interval graphs.
We show in Theorem \ref{thm:FPTBubbles} that for every fixed $\alpha$ there is an $\fpt$ algorithm that computes a maximum cardinality cut of $G$ 
given a bubble partition with independence number $\alpha$ of it.
Since for mixed unit interval graphs, a bubble partition with independence number 1 can be found in polynomial time, 
this implies an $\fpt$ algorithm on the same parameter for this family of graphs.

In Section \ref{sec:CliqueWidth}, we extend the scope of this $\fpt$ algorithm to a wider domain. 
To this purpose the $(\alpha, \beta, \delta)$-clique-width decomposition of a graph is a clique-width decomposition having characteristics bounded by the parameters $\alpha, \beta, \delta$.
In Theorem \ref{thm:fpt_cliquewidth}, we show that the $\xp$ algorithm for $\maxcut$ when parameterized by clique-width presented in \cite{fomin2010algorithmic} 
runs in  $\fpt$ time when parameterized by the smallest width of an $(\alpha,\beta,\delta)$-clique-width decomposition, denoted by $\cw_{\alpha,\beta,\delta}(G)$. 
We also show (in Lemma \ref{lem:cliquewidth_and_bubblepartiton_relation}) that a bubble partition with independence number $\alpha$ can be used to find an $(\alpha, 1, 1)$-clique-width decomposition 
of similar width. Therefore, in some sense, the main result of Section \ref{sec:CliqueWidth} generalizes the main result of Section \ref{sec:BubblePartition}.

Denoting by $\BW_{\alpha}$ the class of graphs having a bubble partition with independence number $\alpha$,
we present structural results relating $\BW_{1}$ to the classes of interval, mixed unit interval, chordal, co-bipartite and split graphs.

We conclude in Section 5 with several open questions about $\fpt$ algorithms for {\sc MaxCut} and related parameters introduced in this paper, both in general and in specific graph classes. 

\section{Preliminaries}\label{sec:prelim}
\runningtitle{Graph notations and terms}
Given a simple graph (, i.e., with no loops or parallel edges) $G=(V(G),E(G))$ and a vertex $v$ of $G$,
$N_G(v)$ denotes the set of neighbors of $v$ in $G$. 
Two adjacent (resp. non-adjacent) vertices $u,v$ of $G$ are \emph{twins} (resp. \emph{false twins}) if $N_G(u) \setminus \set{v} = N_G(v) \setminus \set{u}$. 
For a graph $G$ and $U \subseteq V(G)$, we denote by $G[U]$ the subgraph of $G$ induced by $U$, and $G \setminus U \defined G[V(G) \setminus U]$. 
For a singleton $\set{x}$ and a set $Y$, $Y + x \defined Y \cup \set{x}$ and $Y - x \defined Y \setminus \set{x}$. 
A vertex set $U \subseteq V(G)$ is a \emph{clique} (resp. \emph{independent set}) (of $G$) if every pair of vertices in $U$ is adjacent (resp. non-adjacent). 
We denote by $\alpha(G)$ the maximum size of an independent set of a graph $G$.
We refer the reader to \cite{D12} for general notation and terminology regarding graphs.

\runningtitle{Some graph classes}
A graph is \emph{chordal} if it does not contain holes, i.e., induced cycles of four or more vertices.
It is known that a graph $G$ is chordal if and only if it is the vertex-intersection graph of subtrees of a tree, 
i.e., there exists a tree $T$ and subtrees $T_1, \ldots, T_n$ of it such that $v_i$ and $v_j$ are adjacent in $G$ if and only if $T_i$ and $T_j$ have a common vertex \cite{Gavril74}.
A graph $G$ is \emph{interval} if it is the intersection graph of intervals on a straight line.
The subtree intersection characterization of chordal graphs implies that interval graphs are chordal.
An interval graph is \emph{proper} (resp. \emph{unit}) if it has an interval representation such that no interval properly contains another (resp. every interval has unit length). 
It is known that the class of proper interval graphs is equivalent to the class of unit interval graphs \cite{Bogart199921}.
However, if one is allowed to use a mixture of open and closed unit intervals in the representation, richer families of graphs are obtained.
This can be easily demonstrated by the set $\set{[0,1],(1,2),[2,3], [1,2]}$ of unit intervals that represent a claw which is not a proper interval graph.
The most general such family is the family of \emph{mixed unit interval graphs} of the intersection graphs of unit intervals 
where each interval can be open, closed, or open at one end and closed at the other.

\runningtitle{Cuts}
We denote a cut of a graph $G$ by one of the subsets of the partition. 
$E(S,\bar{S})$ denotes the \emph{cut-set} of $S$, i.e. the set of the edges of $G$ with exactly one endpoint in $S$, and $\cutsize{S} \defined \abs{E(S,\bar{S})}$ is termed the \emph{cut size} of $S$.
A \emph{maximum cut} of $G$ is one having the biggest cut size among all cuts of $G$. 
We refer to this size as the \emph{maximum cut size} of $G$. 
Clearly, $S$ and $\bar{S}$ are dual; we thus can replace $S$ by $\bar{S}$ and $\bar{S}$ by $S$ everywhere. 
In particular, $E(S,\bar{S})=E(\bar{S},S)$, and $\cutsize{S}=\cutsize{\bar{S}}$.

\runningtitle{Parameterized Complexity}
A \emph{parameterized problem} is a decision problem each instance of which is a pair $(I,k)$ where $k$ is a number that is termed the \emph{parameter} of the instance. 
An algorithm that decides a parameterized problem $\Pi$ is an $\fpt$ (resp. $\xp$) algorithm if its running time is bounded by $f(k) \cdot |I|^c$ (resp. $f(k) \cdot |I|^{g(k)}$)
for some computable functions $f, g$ and some constant $c$.  
The class {\fpt} (resp. {\xp}) is the class of all parameterized problems for which an {\fpt} (resp. {\xp}) algorithm exists.
Clearly, $\fpt \subseteq \xp$.
A parameterized problem that is in {\fpt} is termed \emph{fixed-parameter tractable}.
The notation $\bigoh^*$ is used to omit polynomial factors. 
For instance, for an $\fpt$ algorithm of time complexity $\bigoh(f(k) \cdot |I|^c)$ for some constant $c$, we omit the polynomial factor of $|I|^c$ and 
say that the time complexity of the algorithm is $\bigoh^*(f(k))$.
We refer the reader to~\cite{DF13,CyganFKLMPPS15} for basic background on parameterized complexity.

\runningtitle{Bubble models}
A \emph{2-dimensional bubbles model} $\BB$ for a finite non-empty set $A$ is a 2-dimensional arrangement of bubbles $\set{B_{i,j}~|~j \in [k], i \in [r_j]}$ for some positive integers $k, r_1, \ldots ,r_k$, such that $\BB$ is a \emph{near-partition} of $A$.
That is, $A = \cup \BB$ and the sets $B_{i,j}$ are pairwise disjoint, allowing for the possibility of $B_{i,j}=\emptyset$ for arbitrarily many pairs $i,j$. 
For an element $a \in A$ we denote by $i(a)$ and $j(a)$ the unique indices such that $a \in B_{i(a),j(a)}$.
Given a bubble model $\BB$, the graph $G(\BB)$ has $\cup \BB$ as its vertex set.
Two vertices $u,v$ are adjacent (in $G(\BB)$) if and only if
$j(u)=j(v)$ or, $j(u)=j(v)+1$ and $i(u) < i(v)$.
We say that $\BB$ is a \emph{bubble model} for $G(\BB)$. 
Observe that every \emph{bubble} $B \in \BB$ is a set of twins.
A \emph{compact representation} for a bubble model is an array of \emph{columns} each of which contains a list of non-empty bubbles given by their row numbers and their vertices.

\begin{theorem}\cite{HMP09}\label{thm:Bubbles}
\begin{enumerate}[i)]
\item A graph is proper interval if and only if it has a bubble model.
\item A compact representation of a bubble model for a proper interval graph on $n$ vertices can be computed in linear time.
\end{enumerate}
\end{theorem}

Recently an extended bubble model for mixed unit interval graphs is introduced (\cite{Kratochvil-UBubbleModel-2020}). 
Such a bubble model can be obtained by partitioning every bubble $B_{i,j}$ of the bubble model of a corresponding unit interval graph into four \emph{quadrants} $B^1_{i,j},B^2_{i,j},B^3_{i,j},B^4_{i,j}$ 
where every quadrant contains intervals of one type (open at both ends, closed at both ends, and so on).
There is an edge between vertices $u,v$ of when $j(u)=j(v)+1$ and $i(u) = i(v)$ if they belong to appropriate quadrants.

We denote by $p(G)$ the maximum number of non-empty bubbles in a column of the bubble model (resp. extended bubble model) $\BB$ of a proper (resp. mixed unit) interval graph $G$. 

In this work, we use the term bubble as a maximal set of twins and extend the scope of this definition to general graphs, not restricted to proper interval graphs. 
Whenever an ambiguity arises, we use the adjective 2-dimensional for the bubble model of a proper interval graph.

\runningtitle{Clique-width}
\emph{Clique-width} of a graph $G$ is the minimum number of labels to construct $G$ by using the following four operations defined on vertex-labeled graphs:
\begin{enumerate}
    \item The operation
    $\ell(v)$ returns a graph with one vertex $v$ labeled $\ell$.
    \item The disjoint union $G \cup G'$ of two vertex-disjoint labeled graphs $G$ and $G'$ is the graph $(V(G) \cup V(G'), E(G) \cup E(G'))$ and every vertex in $V(G) \cup V(G')$ preserves its original label.
    \item The graph $\eta_{i,j}(G)$ is obtained from the graph $G$ by connecting all the vertices labeled $i$ with all the vertices labeled $j$.
    \item The graph $\rho_{i \rightarrow j}(G)$ is obtained from the graph $G$ by replacing all labels $i$ with $j$.
\end{enumerate}

A \emph{clique-width decomposition} of a graph is a rooted binary tree that represents an expression involving the above four operations.
The graph $G_t$ corresponding to a node $t$ of $T$ is the value of the expression represented by the subtree of $T$ rooted at $t$.
Let $\LL(T)$ be the set of labels used by $T$.
The \emph{width} $w(T)$ of a decomposition $T$ is the number $\abs{\LL(T)}$ of labels it uses.
The \emph{clique-width} $\cw(G)$ of an (unlabeled) graph $G$ is the smallest width of an expression whose value is $G$ (with some labeling function).
For a label $\ell \in \LL$, we denote by $V_\ell$ the set of vertices labeled $\ell$ and by $V_{t,\ell}$ the set of vertices of $G_t$ labeled $\ell$.
We denote by $\VV_t$ the partition $\set{V_{t,\ell} \mid \ell \in \LL(T)}$ of $V(G_t)$.

\runningtitle{Decomposition by clique separators}
The concept of decomposition by clique separators is introduced by Tarjan \cite{Tarjan-CliqueSeparators}.
If $G$ is a connected graph and $K$ a clique of $G$ such that $G \setminus K$ is disconnected with connected components $V_1, V_2, \ldots, V_k$ then we decompose
$G$ into $k$ subgraphs $G[K \cup G_1], G[K \cup G_2], \ldots, G[K \cup G_k]$. 
By continuing recursively for every subgraph until a subgraph does not contain a clique separator, we obtain a decomposition of $G$.
This decomposition can be modeled by a tree $T$ where an internal node of $T$ represents a clique of $G$ 
and the leaves of $T$ represent subgraphs of $G$ termed \emph{atoms} that do not contain clique separators.
Given any graph, such a decomposition can be found in polynomial-time.

\section{Bubble Partitions}\label{sec:BubblePartition}
Following the definition of 2-dimensional bubble representations of proper interval graphs,
we term {\em bubble} a maximal set of twins.
Given a graph $G$, we denote by $G^-$ the graph obtained by contracting every bubble of $G$ to a single vertex.

A \emph{bubble partition} of a graph $G$ is a partition $\VV = \set{V_1, \ldots V_k}$ of $V(G)$
such that every $V_i \in \VV$ is a union of bubbles and the graph obtained from $G$ by contracting every set $V_i$ to a single vertex is a tree $T(\VV)$.
Note that a bubble partition $\VV=\set{V_1, \ldots, V_k}$ of $G$ corresponds to a partition $\VV^-=\set{V^-_1, \ldots, V^-_k}$ of $V(G^-)$.

A bubble partition always exists, since $\set{V}$ is a partition whose contraction results in a (trivial) tree.
The \emph{independence number} $\alpha(\VV)$ of a bubble partition $\VV$ is $\max \set{ \alpha(V_i) \mid V_i \in \VV}$,
and the \emph{width} $w(\VV)$ of $\VV$ is $\max \set{\abs{V_i^-} \mid V_i \in \VV}$, the largest number of bubbles in a set of $\VV$.
The \emph{$\alpha$-bubble width} $\bw_\alpha(G)$ of $G$ is the smallest width of a bubble partition $\VV$ such that $\alpha(\VV) \leq \alpha$ 
(and $\infty$ if no such partition exists).

Given a cut $S$ and a set $U$ of (false or true) twin vertices, we denote by $S^+(U,\ell)$ the cut obtained by adding $\ell$ vertices of $U \sm S$ to $S$ 
where $\ell \in [\abs{U \sm S}]$. 
Similarly, we denote by $S^-(U,t)$ the cut obtained by removing $t$ vertices of $U$ from $S$,
where $t \in [\abs{U \cap S}]$. 

\begin{obs}\label{obs:AddRemoveFalseTwins}
Let $U$ be an independent set of pairwise (false) twin vertices of a graph $G$, and $S$ a cut of $G$. Then
\[
\cutsize{S^+(U,\ell)} = \cutsize{S} + \ell \cdot (\abs{N(U) \sm S} - \abs{N(U) \cap S}). 
\]
\end{obs}

For a set $U$ of pairwise false twins, the marginal contribution of $U$ to $S$ is defined as $\delta(U,S) \defined \abs{N(U) \sm S} - \abs{N(U) \cap S}$. 
Then, $\cutsize{S^+(U,\ell)} = \cutsize{S} + \ell \cdot \delta(U,S)$ whenever $U$ is a set of pairwise false twins. 
Furthermore, we have
\[
\delta(U,\bar{S}) = \abs{N(U) \setminus \bar{S}} - \abs{N(U) \cap \bar{S}} = \abs{N(U) \cap S} - \abs{N(U) \setminus S} = - \delta(U,S),
\] 
and
\[
\cutsize{S^-(U,\ell)} = \cutsize{\bar{S}^+(U,\ell)} =  \cutsize{\bar{S}} + \ell \delta(U,\bar{S}) = \cutsize{S} + \ell \cdot \delta(U,\bar{S}) = \cutsize{S} - \ell \cdot \delta(U,S).
\]

Given two adjacent bubbles $B,B'$ and a cut $S$ of a graph $G$,
we denote by $S(B,B',\ell)$ the cut obtained from $S$ by adding to it $\ell$ vertices of $B \sm S$ and removing from it $\ell$ vertices of $B' \cap S$,
provided that $\ell \leq \min \set{\abs{B \setminus S}, \abs{B' \cap S} }$. 
Note that $B \cup B'$ is a clique.
Since the number of edges of $E[S,\bar{S}]$ in this clique is not affected by this operation, applying Observation \ref{obs:AddRemoveFalseTwins} twice, we get the following.

\begin{obs}\label{obs:TransferTwins}
\[
\cutsize{S(B,B',\ell)} - \cutsize{S} = \ell \cdot (\delta(B,S) - \delta(B',S)) = - (\cutsize{S(B',B,\ell)} - \cutsize{S}). 
\]
\end{obs}

Two sets $A$ and $B$ are \emph{crossing} (or $A$ crosses $B$ and vice versa) if their intersection is non-empty and none of them is a subset of the other.
A cut $S$ of a graph $G$ is \emph{tight} if the set of bubbles of $G$ crossed by $S$ corresponds to an independent set of $G^-$.

\begin{theorem}\label{thm:TightCut}
Every graph has a tight maximum cut.
\end{theorem}

\begin{proof}
Suppose that the statement does not hold.
Let $S$ be a maximum cut of $G$ that is not tight, 
such that the number of bubbles that $S$ crosses is smallest possible.
Then $S$ crosses at least two adjacent bubbles $B,B'$ of $G$.
We recall that each bubble consists of twin vertices. 
Therefore, $B \cup B'$ is a clique of $G$.
By Observation \ref{obs:TransferTwins}, at least one of $\cutsize{S(B,B',\ell)} - \cutsize{S}$ and $(\cutsize{S(B',B, \ell)} - \cutsize{S})$ is non-negative for every feasible $\ell \geq 0$. 
In the sequel we assume that $\cutsize{S(B,B',\ell)} \geq \cutsize{S}$ with the other case being symmetric. 
Let $\ell = \min \set{ \abs{B \sm S}, \abs{B' \cap S} }$, and note that $\ell>0$.
Then $S(B,B',\ell)$ is a maximum cut that does not cross at least one of $B$ and $B'$. 
Since the intersection of $S$ with other bubbles is not affected, the number of bubbles that $S$ crosses is reduced by one,
contradicting the way $S$ is chosen.
\end{proof}

Let $U$ be a union of bubbles. 
A \emph{configuration} of $U$ for some tight cut $S$ is an encoding $\gamma(S,U)$ of $U \cap S$, defined as follows.
$\gamma(S,U)$ is a triple $(\II, \vect{s}, \bar{U}^-)$ where $\II$ is a (possibly empty) independent set of $G[U]^-$ that indicates the set of bubbles of $U$ that $S$ crosses, $\vect{s}$ is a vector of (non-negative) integers indexed by the elements of $\II$. 
For a bubble $B \in U$ corresponding to a vertex of $\II$, the number $\vect{s}_B \in [\abs{B}-1]$ indicates the number of vertices of $B$ in $S$. 
The set $\bar{U}^- \subseteq U^- \setminus \II$ indicates the set of bubbles that are completely in $S$.
Therefore, in the sequel we denote $S \cap U$ also as $\gamma(S,U)$, interchangeably.
We also denote $\Gamma(U) = \set{\gamma(S,U) \mid S \textrm{~is tight}}$.

The first entry of $\gamma(S,U)$ can be chosen in at most $\sum_{i \leq \alpha(G[U]^-)} {\abs{U^-} \choose i} < \abs{U^-}^{\alpha(G[U]^-)+1}$ different ways.
The second entry can be chosen in at most $\abs{U}^{\alpha(G[U]^-)}$ different ways, 
and the last entry can be chosen in at most $2^{\abs{U^-}}$ ways.
Therefore, 
\[
\abs{\Gamma(U)} \leq \abs{U}^{2\alpha(G[U]^-)+1} 2^{\abs{U^-}} = \abs{U}^{2\alpha(G[U])+1} 2^{\abs{U^-}}.
\]

\begin{theorem}\label{thm:FPTBubbles}
Given a bubble partition $\VV$ of a graph $G$, a maximum cut of $G$ can be computed in time $\bigoh^* \left( \abs{V(G)}^{4\alpha(\VV)} 4^{w(\VV)}  \right)$.
\end{theorem}
\begin{proof}
Consider the tree $T=T(\VV)$ of the bubble partition $\VV$ with an arbitrarily chosen root $r$.
We denote by $\CC(t)$ the set of children of a node $t$ in $T$, 
and $V_t \in \VV$ is the set of $\VV$ from the contraction of which $t$ is obtained.
Let $T_t$ be the subtree of $T$ induced by $t$ and all of its descendants. 
Accordingly, $G_t$ denotes the subgraph of $G$ induced by all the vertices of $G$ represented by the nodes of $T_t$.
We process the nodes of $T$ from the bottom to the top and compute a set of best cuts of $G_t$, namely one cut for each possible configuration of $V_t$.
We terminate after the root $r$ is processed, and choose a configuration in $\Gamma(V_r)$ leading to a maximum cut of $G$.

For a node $t$ of $T$ and a configuration $\gamma \in \Gamma(V_t)$,
we denote by $OPT_t(\gamma)$ the maximum size of a (tight) cut $S_t$ of $G_t$ such that $\gamma$ encodes $S \cap V_t$, i.e.,
\[
OPT_t(\gamma) = \max \set{\cutsize{S_t} \mid \gamma(S,V_t) = \gamma}.
\]
By Theorem \ref{thm:TightCut}, the maximum cut size of $G$ is $\max_{\gamma \in \Gamma(V_r)} OPT_r(\gamma)$.
In the sequel we show how to compute the values $OPT_t(\gamma)$ from the values $OPT_{t'}(\gamma')$ of the children $t'$ of $t$.

Let $S_t$ be a tight cut of $G_t$, and for $t' \in \CC(t)$, 
let $S_{t'}$ denote the cut induced by $S_t$ on $G_{t'}$.
Denote by $E_{t,t'}(S_t, \bar{S_t})$ the set of edges between $V_t$ and $V_t'$ that are separated by $S_t$, 
i.e. $E_{t,t'}(S_t, \bar{S_t}) = E(G) \cap \left( (S_t \cap V_t) \times (\bar{S_t} \cap V_{t'}) \cup (\bar{S_t} \cap V_t) \times (S_t \cap V_{t'}) \right)$.
Since the vertices of $V_t$ are adjacent (in $G_t$) only to vertices of $\cup_{t' \in \CC(t)} V_{t'}$,
we have
\[
\cutsize{S_t} = \cutsize{S_t \cap V_t}  + \sum_{t' \in \CC(t)} \left( \abs{E_{t,t'}(S_t, \bar{S_t})} + \cutsize{S_{t'}} \right).
\]
Therefore,
\[
OPT_t(\gamma) =  \cutsize{\gamma}  + \sum_{t' \in \CC(t)} \max_{\gamma' \in \Gamma(V_{t'})} \left( \abs{E_{t,t'}(\gamma, \gamma')} + OPT_{t'}(\gamma') \right).
\]
Clearly, $E_{t,t'}(\gamma, \gamma')$ can be computed in $\bigoh(\abs{E(G)})$ time and
$OPT_t(\gamma)$ can be computed in time 
\[
\bigoh^* \left( \sum_{t' \in \CC(t)} \abs{\Gamma(V_{t'})} \right) \leq \bigoh^* \left( \sum_{t' \in \CC(t)} \abs{V_{t'}}^{2\alpha(V_{t'})+1} 2^{\abs{V^-_{t'}}}  \right) \leq \bigoh^* \left( \abs{V(G)}^{2\alpha(\VV)} 2^{w(\VV)}  \right).
\]
For every node $t$, we compute $\abs{\Gamma(V_t)} = \bigoh^* (\abs{V(G)}^{2\alpha(\VV)} 2^{w(\VV)})$ values of $OPT_t$.
Therefore, the running time of the algorithm is $\bigoh^* \left( \abs{V(G)}^{4\alpha(\VV)} 4^{w(\VV)}  \right)$.
\end{proof}

\begin{remark}
Optimal configurations of bubbles of an independent set can be computed independently of each other, once the configuration of other bubbles is given.
Using this observation, the above running time can be improved to  $\bigoh^* \left( \abs{V(G)}^{2\alpha(\VV)} 4^{w(\VV)}  \right)$.
\end{remark}

Denoting by $\BW_\alpha$ the class of graphs $G$ such that $\bw_\alpha(G) < \infty$, we formulate the following corollary of Theorem \ref{thm:FPTBubbles}.

\begin{corollary}\label{coro:fptbwalpha}
For every $\alpha > 0$, there is an $\fpt$ algorithm for $\maxcut$ for $\BW_\alpha$ when parameterized by $\bw_\alpha(G)$
provided that a bubble partition $\VV$ of width $\bw_\alpha(G)$ can be found in time $\bigoh^*(f(\bw_\alpha(G)))$ for some computable function $f$.
\end{corollary}

Recall that $p(G)$ denotes the maximum number of non-empty bubbles in a column of the 2-dimensional bubble model $\BB$ of a proper (or mixed unit) interval graph $G$. 

\begin{corollary}\label{coro:FPTProperInterval}
There is an $\fpt$ algorithm for $\maxcut$ in mixed unit interval graphs when parameterized by $p(G)$. 
Moreover, $\bw_1(G) \leq p(G) $ whenever $G$ is a mixed unit interval graph.
\end{corollary}
\begin{proof}
Let $G$ be a mixed unit interval graph.
The 2-dimensional bubble-representation $\BB$ of $G$ can be computed in polynomial time. 
Let $V_j$ be a column of $\BB$, i.e., $V_j=\bigcup_{i=1}^{r_j} \bigcup_{q=1}^4 B^q_{i,j}$,
and consider the partition $\VV=\set{V_j \mid j \in [k]}$.
Every set $V_j \in \VV$ is a clique and also a union of bubbles.
Moreover, the graph obtained from the contraction of every $V_j$ to a single vertex is a path. 
Therefore, $\VV$ is a bubble partition with $\alpha(\VV)=1$ and $w(\VV)=p(G)$.
By Theorem \ref{thm:FPTBubbles}, there is an algorithm for $\maxcut$ that runs in time $\bigoh^* \left( \abs{V(G)}^4 4^{w(\VV)}  \right)= \bigoh^* \left( 4^{p(G)}  \right)$.
\end{proof}

We conclude this section by relating $\BW_1$ to some known graph classes.
By the proof of Corollary \ref{coro:FPTProperInterval}, $\BW_1$ contains the class of mixed unit interval graphs.
It is easy to see that $\BW_1$ contains also the classes of split graphs and co-bipartite graphs.
Therefore, we have the following.
\begin{obs}
$
\textsf{Split} \cup \textsf{Co-Bipartite} \cup \textsf{MixedUnitInterval} \subseteq \BW_1.
$
\end{obs}

Clearly, $G \in \BW_1$ if and only if $G$ has a bubble partition where each set is a clique. 
At first glance, such a bubble partition seems to be a special case of decomposition by clique separators. 
A result of Dirac \cite{Dirac-RijidCircuit} implies that a graph is chordal if and only if it has a decomposition by clique separators the atoms of which are cliques.
Given these facts it is natural to investigate the relationship between the class $\BW_1$ and the class of chordal graphs.

\begin{theorem}\label{thm:bw1_crosses_chorldal_graphs}
$\BW_1$ crosses both classes of chordal and interval graphs.
\end{theorem}
\begin{proof}
Since every clique is both chordal, interval and $\BW_1$, the intersection of these classes is non-empty. 
Moreover, a $C_4$ (being co-bipartite) is in $\BW_1$ but not chordal.
It remains to show that that there is an interval graph which is not $\BW_1$.

Consider the graph $G$ on 8 vertices obtained by adding a universal vertex $v_0$ to a path $P$ on 7 vertices $v_1, \ldots, v_7$ where the vertices are numbered according to their order on $P$.
It is trivial to construct an interval representation for $G$. 
We claim that $G \notin \BW_1$. 
Assume for a contradiction that $G \in \BW_1$, 
and let $\VV=\set{V_0, V_1, \ldots, V_k}$ be a bubble partition of $G$ such that $\alpha(V_i)=1$ for $i=0,\ldots,k$, i.e., every $V_i$ is a clique. 
Note that each bubble consists in a single vertex as $G$ is twin-free. 
Assume without loss of generality that $v_0 \in V_0$. 
Then, the node 0 of $T(\VV)$ corresponding to $V_0$ is adjacent to every other node. 
In other words, $T(\VV)$ is a star with center 0 and leaves $1,\ldots,k$. 
Since every $V_i$ is a clique, $V_0$ contains at most two vertices of $P$ (in addition to $v_0$). 
Then $P \setminus V_0$ has at least 5 vertices and at most two connected components,
Implying that $P \setminus V_0$ has a connected component with at least three vertices.
This yields two adjacent nodes in $T(\VV)$, contradicting that $T(\VV)$ is a star with center 0.
\end{proof}

\section{Clique-width Decompositions}\label{sec:CliqueWidth}
The $\maxcut$ problem can be solved in polynomial time for graphs with bounded clique-width.  
However, an $\fpt$ algorithm for $\maxcut$ when parameterized by the clique-width of the input graph is impossible under $\ETH$ \cite{fomin2010algorithmic}.
In this section, we consider clique-width decompositions with special properties, and the behaviour of the $\maxcut$ algorithm under such decompositions.
We show that these decompositions in some sense extend bubble partitions by giving a construction of a clique-width decomposition whose width is a constant factor away of the width of a given bubble partition.

We start with properties of clique-width decompositions that we will assume without loss of generality. 
Let $r$ be the root of a clique-width decomposition $T$ of $G$ (i.e., $G_r=G$),  $t$ be a node of $T$ with parent $t'$.

\begin{itemize}
    \item If $t'$ is a union node then $G_t$ is an induced subgraph of $G$. 
    Indeed, if this is not the case, there are pairs of sets $V_{t,\ell}, V_{t,\ell'}$ such that the vertices of $V_{t,\ell}$ and $V_{t,\ell'}$ are adjacent in $G$ but not adjacent in $G_t$.
    Then, we can insert an $\eta_{\ell,\ell'}$ node between $t$ and $t'$ for every such pair $\ell,\ell'$. 
    This modification does not affect the width of the decomposition.
    \item If two vertices $u,v$ are twins in $G$ and $u \in V(G_t)$, then $v \in V(G_t)$. 
    Moreover, $u$ and $v$ have the same label $\ell$ in $G_t$.
    If this is not the case, we can remove from $T$ the node $\ell'(v)$ (and every parent node with one child), 
    and replace the expression $\ell(u)$ by the expression $\eta_{\ell,\ell}(\ell(u) \cup \ell(v))$. 
    Therefore,
    \item $V(G_t), V_{t,\ell}$ are non-crossing sets, and $\VV_t$ is a non-crossing partition of $V(G_t)$ for every $t$ and $\ell$.
    \item $u,v \in V(G_t)$ are twins in $G$ if and only if they are twins in $G_t$ and they have the same label in $G_t$.
    \item If a cut $S$ is tight then the cut $S_t$ that $S$ induces on $G_t$ is tight.
\end{itemize}

A clique-width decomposition $T$ is \emph{an $(\alpha,\beta,\delta)$-clique-width decomposition} if for every node $t$ of $T$,
there exists a set $L_t \subseteq \LL(T)$ of at most $\delta$ labels such that
\begin{itemize}
    \item the independence number of $G[\bigcup_{\ell \notin L_t} V_{t,\ell}]$ is at most $\alpha$, and
    \item the number of bubbles $\abs{V_{t,\ell}^-}$ of $V_{t,\ell}$ is at most $\beta$ whenever $\ell \notin L_t$.
\end{itemize}

We now analyze the running time of the algorithm in \cite{fomin2010algorithmic} that solves $\maxcut$ when provided with an $(\alpha,\beta,\delta)$-clique-width decomposition $T$ of the input graph $G$.
The algorithm presented in \cite{fomin2010algorithmic}  is based on the following observation.
For every graph $G_t$ and every label $\ell$ the vertices of $V_{t,\ell}$ are identical with respect to vertices not in $G_t$.
Therefore, when processing $T$ in a bottom-up fashion and $t$ is the current node, two cuts $S$  and $S'$ such that $\abs{S \cap V_{t,\ell}}=\abs{S' \cap V_{t,\ell}}$ are identical with respect to vertices of $G \setminus G_t$.
For every node $t$ of $T$ and every vector $\vect{s} \in {\mathbb N}^{w(T)}$ such that $0 \leq s_\ell \leq \abs{V_{t,\ell}}$ for every $\ell \in \LL(T)$ the algorithm computes the maximum cut size among all cuts $S$ such that $\abs{S \cap V_{t,\ell}}=s_\ell$.
The running time of the algorithm is dominated by the computation at union nodes in which, in order to compute the result for a vector $\vect{s'}$ of a parent node $t'$, 
the algorithm considers all the vectors  of $\vect{s}$ of one of the children and
for each such vector, the vector $\vect{s'}-\vect{s}$ for the other child.
Since the number of vectors $\vect{s}$ is bounded by $n^{w(T)}$ it follows that the running time of the algorithm is $n^{\bigoh(w(T))}$.

We now improve this upper bound using the above observations. 
Let $S$ be a tight maximum cut of $G$.
Then the cut $S_t$ that $S$ induces on $G_t$ is also tight, for every node $t$ of $T$.
Therefore, it suffices to consider only cuts that are tight in $G_t$.
More precisely, it is sufficient to compute the results only for vectors that result from a tight cut.

To guess a tight cut $S_t$ we first guess an independent set $\II$ of $G[\bigcup_{i \notin L_t} V_{t,i}]$ in one of the at most $(\beta \cdot (w(T)-\delta))^{\alpha+1}$ ways.
Then, for every bubble $B$ that intersects $\II$, we guess the number of vertices in $S_t \cap B$.
This can be done in at most $\abs{V(G_t)}^\alpha$ different ways.
For every bubble that is a) not labeled with a label from $L_t$, 
and b) does not intersect $\II$ 
we guess whether or not it is contained in $S_t$.
This can be done in at most $2^{\beta \cdot (w(T)-\delta)}$ ways.
Finally, we guess the number of vertices of $S_t \cap V_{t,\ell}$ for every label $\ell \in L_t$.
This can be done in at most $\abs{V(G_t)}^\delta$ ways.
We conclude that the number of vectors $\vect{s}$ to consider is at most
\[
(\beta \cdot (w(T)-\delta))^{\alpha+1} \abs{V(G_t)}^\alpha 2^{\beta \cdot (w(T)-\delta)} \abs{V(G_t)}^\delta=\bigoh(\abs{V(G_t)}^{\alpha+\delta} (\beta \cdot w(T))^{\bigoh(\alpha)} 2^{\beta \cdot w(T)}).
\]

Let $\cw_{\alpha,\beta,\delta}(G)$ be the smallest width of an $(\alpha,\beta,\delta)$-clique-width decomposition of 
$G$ (and $\infty$ if no such decomposition exists).
We conclude that for every $\alpha,\beta,\delta > 0$ the running time of the algorithm is $\bigoh^*((\beta \cdot \cw_{\alpha,\beta,\delta}(G))^{\bigoh(\alpha)} 2^{\beta \cdot \cw_{\alpha,\beta,\delta}(G)})$.
\begin{theorem}\label{thm:fpt_cliquewidth}
The $\maxcut$ problem when parameterized by $\cw_{\alpha,\beta,\delta}(G)$ is in $\fpt$ for every $\alpha,\beta,\delta>0$.
\end{theorem}
We now present the following lemma which implies that Corollary \ref{coro:FPTProperInterval} in fact extends Theorem \ref{thm:FPTBubbles}.https://www.overleaf.com/project/5c1537ee4b64b64084fba2f7

\begin{lemma}\label{lem:cliquewidth_and_bubblepartiton_relation}
\[
\cw_{\alpha,1,1}(G) \leq 2 \bw_\alpha(G) + 1.
\]
Moreover, given a bubble partition $\VV$ of $G$, one can find an $(\alpha(\VV),1,1))$-clique-width decomposition of $G$ of width $2 w(\VV) + 1$ in polynomial time.
\end{lemma}
\begin{proof}
Let $T=T(\VV)$ rooted at an arbitrary node $r$, $\alpha=\alpha(\VV)$, $w=w(\VV)$.
For every node $t$ of $T$, let $G_t$ be the subgraph of $G$ induced by the set of vertices that are in the subtree of $T$ rooted at $t$.
Let $\LL=\set{\ell_1, \ldots, \ell_w}$, $\LL'=\set{\ell'_1, \ldots, \ell'_w}$ two label sets.
For every node $t$ of $T$, we will construct an expression whose value is $G_t$ and uses labels from the collection $\LL \cup \LL' \cup \set{0}$. 
In particular, we will obtain an expression for $G=G_r$.
Our construction will guarantee that $\abs{V_{t,\ell}^-}=1$ for every $\ell \in \LL \cup \LL'$ and
$\alpha(G_t \setminus V_{t,0}) \leq \alpha$.
Thus, the expression for $G_r$ will be an $(\alpha,1,1)$ clique-width decomposition of $G$ of width $\abs{\LL \cup \LL' \cup \set{0}}=2w+1$.

Let $B(k,\ell)$ be the expression $\eta_{\ell,\ell}\left( \bigcup_{i=1}^k \ell(i) \right)$ whose value is a $k$-clique where every vertex is labelled by $\ell$.
For a set $X$ of pairs of labels, we denote by $\eta_X$ a path consisting of nodes $\eta_{\ell,\ell'}$, one for every pair $(\ell,\ell') \in X$.
If $X \subseteq A \times B$ for two disjoint sets $A,B$ we denote by $\rho_X$ a path of $\rho_{i \rightarrow j}$ nodes, one for every pair $(i,j) \in X$.
Note that $\eta_{i,j}$ operations are commutative and the operations $\rho_{i \rightarrow j}$, $\rho_{i' \rightarrow j'}$ are commutative if $\set{i,i'} \cap \set{j,j'} = \emptyset$, i.e. the above definitions are non-ambiguous.
In particular, we denote by $\rho_{\ell \rightarrow \ell'}$ the operation of relabeling all nodes labeled $\ell_i$ with $\ell'_i$ and
by $\rho_{\ell' \rightarrow 0}$ the operation of relabeling all nodes labeled $\ell'_i$ with $0$ for every $i \in [w]$.

We now describe our construction having the extra property that the graph $G_t$ is labeled with labels from $\LL \cup \set{0}$ and every vertex not in $V_t$ is labeled 0, for every node $t$ of $T$.
For a leaf $t$ of $T$ the graph $G_t$ is the subgraph of $G$ induced by the vertices of $V_t$.
We have $\alpha(V_t)\leq \alpha$ and $\abs{V_t^-} \leq w$, i.e. $V_t$ is a union of at most $w$ bubbles.
Let $k_i$ be the number of vertices in bubble $i$.
Let $E_t$ be the edge set of $G[V_t^-]$.
Then the value of $\eta_{E_t} (\bigcup_{i \in [w]} B(k_i, \ell_i))$ is $G_t$  and it has all the claimed properties.

Let $t$ be a non-leaf node of $T$ with children $t_1, \ldots, t_k$ and $e_1, \ldots, e_k$ expressions for $G_{t_1}, \ldots, G_{t_k}$ 
each of which uses labels from $\LL \cup \LL' \cup \set{0}$ and their values are graphs where vertices of $V_{t_i}$ are labeled with labels from $\LL$ and the rest are labeled 0.

Informally, we start with a construction for $V_t$ as if $t$ were a leaf 
and for every child $t_{k'}$ of $T$ we
a) relabel all the labels $\LL$ of $G_{t_{k'}}$ by $\LL'$,
b) take the disjoint union with $G_{t_{k'}}$,
c) add the edges between $V_t$ and $V_{t_{k'}}$ as needed,
d) relabel all the labels $\LL'$ by 0.

More specifically, for $k' \in [0,k]$ we construct in an expression $e'_i$ whose value is the subgraph of $G_t$ induced by the vertices $V_t$ and all the vertices in the subtrees of $t_1, \ldots, t_{k'}$ where vertices of $V_t$ are labeled with labels from $\LL$ and the rest are labelled 0. 
Finally, the expression $e'_k$ is an expression for $G_t$ having all the claimed properties.
The expression $e'_0$ is  $\eta_{E_t} (\bigcup_{i \in [w]} B(k_i, \ell_i))$ where $k_i$ is the number of vertices in bubble $i$ of $V_t$.
For $k' \in [k]$, let $E_{k'} \subseteq \LL \times \LL'$ be such that $(i.j') \in E_{k'}$ if and only if there is an edge between vertices labeled $i$ in $V_t$ and vertices labeled $j$ in $V_{t_{k'}}$.
Then $e'_{k'} = \rho_{\LL' \rightarrow 0} (\eta_{E_{k'}} (e'_{i-1} \cup \rho_{\LL \rightarrow \LL'} (e_i)))$. 
\end{proof}

Combining Lemma \ref{lem:cliquewidth_and_bubblepartiton_relation} with Corollary \ref{coro:FPTProperInterval} we get the following corollary.
\begin{corollary}\label{coro:fpt2proper}
A $(1,1,1)$-clique-width decomposition of width $2p(G)+1$ can be computed in polynomial-time whenever $G$ is a mixed unit interval graph.
\end{corollary}

\section{Conclusion and Future Work}\label{sec:conclusion}
In this work, we introduced bubble partitions of graphs and new parameters for clique-width decompositions.
We have shown that the existing $\xp$ algorithm for $\maxcut$ parameterized by clique-width runs in  $\fpt$ time (in the width of the clique-width decomposition) for every fixed values of these parameters.
We have shown that bubble partitions with bounded width and independence number can be used to a find clique-width decomposition with bounded values
of these parameters.
For mixed unit interval graphs such bubble partitions can be found in polynomial time.

Our work can be extended in the following directions. 
The dynamic programming algorithm for bubble partitions can be extended to cases where the partition induces a graph with small tree-width instead of a tree. 
One can extend the "look-ahead" for bubbles, to structures that can be decomposed into a small number of modules.
One can study the time complexity of constructing a bubble partition $\VV$ having particular $\alpha(\VV), \bw(\VV)$ parameters.
Such a study may also be confined to specific graph classes.  
It would be also interesting to use such partitions when dealing with problems other than $\maxcut$. 

A bubble partition with independence number $1$ for mixed unit interval graphs easily follows from earlier work.
The complexity of computing $\bw_\alpha(G)$ in general, or for specific graph classes is a research problem that was out of the scope of this work.

It is known that an $\fpt$ algorithm for $\maxcut$ when parameterized by clique-width is unlikely in general \cite{fomin2010algorithmic}.
The existence of such an algorithm for proper interval graphs is an open question.

There seems to be a close relation between bubble partitions with independence number 1 and decomposition by clique separators whose atoms are cliques.
The latter is known to coincide with the class of chordal graphs.
On the other hand, we have shown that the former class neither includes nor is included in the class of chordal graphs.
The characterization of chordal graphs (and interval graphs) that admit a bubble partition with independence number 1 is an interesting research question too.


\vspace{0.3cm}

\bibliographystyle{abbrv}
\bibliography{Approximation,GraphTheory,Mordo,References}

\begin{thebibliography}{10}

\bibitem{Bodlaender2004}
H.~Bodlaender, C.~de~Figueiredo, M.~Gutierrez, T.~Kloks, and R.~Niedermeier.
\newblock Simple max-cut for split-indifference graphs and graphs with few
  {P}\({}_{\mbox{4}}\)'s.
\newblock In C.~Ribeiro and S.~Martins, editors, {\em Experimental and
  Efficient Algorithms}, volume 3059 of {\em Lecture Notes in Computer
  Science}, pages 87--99. Springer Berlin Heidelberg, 2004.

\bibitem{Bodlaender00onthe}
H.~L. Bodlaender and K.~Jansen.
\newblock On the complexity of the maximum cut problem.
\newblock {\em Nordic Journal of Computing}, 7:14 -- 31, 2000.

\bibitem{bodlaenderKN99MaxCutUnitInterval}
H.~L. Bodlaender, T.~Kloks, and R.~Niedermeier.
\newblock {SIMPLE} {MAX-CUT} for unit interval graphs and graphs with few
  {P}\({}_{\mbox{4}}\)s.
\newblock {\em Electronic Notes in Discrete Mathematics}, 3:19--26, 1999.

\bibitem{Bogart199921}
K.~P. Bogart and D.~B. West.
\newblock A short proof that ‘proper = unit’.
\newblock {\em Discrete Mathematics}, 201(1–3):21 -- 23, 1999.

\bibitem{BES16-MaxCutProperInterval}
A.~Boyac{\i}, T.~Ekim, and M.~Shalom.
\newblock A polynomial-time algorithm for the maximum cardinality cut problem
  in proper interval graphs.
\newblock {\em Information Processing Letters}, 121:29--33, May 2017.

\bibitem{CyganFKLMPPS15}
M.~Cygan, F.~V. Fomin, L.~Kowalik, D.~Lokshtanov, D.~Marx, M.~Pilipczuk,
  M.~Pilipczuk, and S.~Saurabh.
\newblock {\em Parameterized Algorithms}.
\newblock Springer, 2015.

\bibitem{DK2007}
J.~Diaz and M.~Kaminski.
\newblock Max-cut and max-bisection are {NP}-hard on unit disk graphs.
\newblock {\em Theoretical Computer Science}, 377(1–3):271 -- 276, 2007.

\bibitem{D12}
R.~Diestel.
\newblock {\em Graph Theory, 4th Edition}, volume 173 of {\em Graduate texts in
  mathematics}.
\newblock Springer, 2012.

\bibitem{Dirac-RijidCircuit}
G.~A. Dirac.
\newblock On rigid circuit graphs.
\newblock {\em Abh.Math.Semin.Univ.Hambg.}, 25:71 -- 76, 1961.

\bibitem{DF13}
R.~G. Downey and M.~R. Fellows.
\newblock {\em Fundamentals of Parameterized Complexity}.
\newblock Texts in Computer Science. Springer, 2013.

\bibitem{BES15-MaxCut-JOCOSpecialIssue}
T.~Ekim, A.~Boyac{\i}, and M.~Shalom.
\newblock The maximum cardinality cut problem in co-bipartite chain graphs.
\newblock {\em Journal of Combinatorial Optimization}, 35:250--265, Jan 2018.

\bibitem{fomin2010algorithmic}
F.~V. Fomin, P.~A. Golovach, D.~Lokshtanov, and S.~Saurabh.
\newblock Algorithmic lower bounds for problems parameterized by clique-width.
\newblock In {\em Proceedings of the twenty-first annual ACM-SIAM symposium on
  Discrete Algorithms}, pages 493--502. Society for Industrial and Applied
  Mathematics, 2010.

\bibitem{Gavril74}
F.~Gavril.
\newblock The intersection graphs of subtrees in trees are exactly the chordal
  graphs.
\newblock {\em Journal of Combinatorial Theory}, 16:47 -- 56, 1974.

\bibitem{Guruswami1999217}
V.~Guruswami.
\newblock Maximum cut on line and total graphs.
\newblock {\em Discrete Applied Mathematics}, 92(2–3):217 -- 221, 1999.

\bibitem{hadlock1975finding}
F.~Hadlock.
\newblock Finding a maximum cut of a planar graph in polynomial time.
\newblock {\em SIAM Journal on Computing}, 4(3):221--225, 1975.

\bibitem{HMP09}
P.~Heggernes, D.~Meister, and C.~Papadopoulos.
\newblock A new representation of proper interval graphs with an application to
  clique-width.
\newblock {\em Electronic Notes in Discrete Mathematics}, 32:27--34, 2009.

\bibitem{Kratochvil-UBubbleModel-2020}
J.~{Kratochv{\'\i}l}, T.~{Masa{\v{r}}{\'\i}k}, and J.~{Novotn{\'a}}.
\newblock {U-Bubble Model for Mixed Unit Interval Graphs and its Applications:
  The MaxCut Problem Revisited}.
\newblock {\em arXiv e-prints}, page arXiv:2002.08311, Feb. 2020.

\bibitem{Tarjan-CliqueSeparators}
R.~E. Tarjan.
\newblock Decomposition by clique separators.
\newblock {\em Discrete Mathematics}, 55(2):221 -- 232, 1985.

\end{thebibliography}
\end{document}